\documentclass[envcountsame]{llncs}
\usepackage[a4paper]{geometry}
\usepackage{amssymb}
\usepackage{amsmath}
\usepackage{amsfonts}
\usepackage{tikz}
\usepackage{url}
\usepackage{a4wide}
\usepackage{xspace}
\usepackage[T1]{fontenc}		
\usepackage{verbatim}
\usepackage{authblk}
\usepackage{nicefrac}
\usepackage[vlined, ruled]{algorithm2e}
\usepackage{todonotes}

\usepackage{fancybox}
\newcommand{\PbOpt}[3]{%
\begin{center}
  \begin{tabular}{|l|}%
  \hline
    \begin{minipage}[c]{.95\textwidth}
    \smallskip%
      \par\noindent%
      \shadowbox{#1}%
      \par\noindent%
      $\bullet$
      \textbf{\textsf{Input}}: #2%
      \par\noindent%
      $\bullet$
      \textbf{\textsf{Output}}: #3%
      \smallskip%
      \par\noindent%
    \end{minipage}
    \\\hline
  \end{tabular}%
\end{center}
}%

\newcommand{\pSC}{\textsc{max $k$-set cover}\xspace}
\newcommand{\SC}{\textsc{min set cover}\xspace}

\newcommand{\DS}{\textsc{min dominating set}\xspace}
\newcommand{\pSAT}{\textsc{max sat-$k$}\xspace}
\newcommand{\mSAT}{\textsc{max sat}\xspace}
\newcommand{\SAT}{\textsc{sat}\xspace}
\newcommand{\pVC}{\textsc{max $k$-vertex cover}\xspace}
\newcommand{\VC}{\textsc{min vertex cover}\xspace}
\newcommand{\BNTMC}{\textsc{bounded non-deterministic Turing machine computation}\xspace}
\newcommand{\algone}{\texttt{ALG1}}

\newcommand{\kscalg}{\texttt{kSC-ALG}}
\newcommand{\pscschema}{\texttt{pSC-IMPROVED}\xspace}
\newcommand{\algpSATk}{\texttt{pSAT-k}\xspace}

\let\leq\leqslant
\let\geq\geqslant
\let\le\leqslant

\makeatletter
\def\@fnsymbol#1{\ensuremath{\ifcase#1\or (a)\or (b)\or (3)\or (4)\or \S \or * \or
   \mathsection\or \mathparagraph\or \|\or **\or \dagger \or \ddagger \or \dagger\dagger
   \or \ddagger\ddagger \else\@ctrerr\fi}}
\makeatother

\title{\textbf{Parameterized Exact and Approximation Algorithms for Maximum $k$-Set Cover and Related Satisfiability Problems}}
\author{\'Edouard Bonnet$^1$,
Vangelis~Th.~Paschos$^2$, Florian Sikora$^2$}
\institute{Institute for Computer Science and Control,
Hungarian Academy of Sciences (MTA SZTAKI) \\
\email{bonnet.edouard@sztaki.mta.hu} 
\and Universit\'{e} Paris-Dauphine, PSL Research University, CNRS, LAMSADE, Paris, France \\
\email{\{paschos,florian.sikora\}@lamsade.dauphine.fr}}

\begin{document}

\maketitle

\protect\thispagestyle{plain}

\begin{abstract}


Given a family of subsets $\mathcal S$
over a set of elements~$X$ and two integers~$p$ and~$k$,
\pSC{} consists of finding a subfamily~$\mathcal T \subseteq \mathcal S$ of cardinality at most~$k$, covering at least~$p$ elements of~$X$. 
This problem is W[2]-hard when parameterized by~$k$, and FPT when parameterized by $p$. 
We investigate the parameterized approximability of the problem with respect to parameters~$k$ and~$p$.
Then, we show that \pSAT{}, a satisfiability problem generalizing \pSC{}, is also FPT with respect to parameter~$p$.
 
\end{abstract}

\section{Introduction}\label{intro}

In the \pSC{} problem, we are given a family of subsets $\mathcal S=\{S_1,\ldots,S_m\}$ over a set of elements $X=\{x_1,\ldots,x_n\}$, and two integers~$p$ and~$k$. 
The goal is to find a subcollection~$\mathcal T$ of at most~$k$ subsets that covers at least~$p$ elements. In what follows, we make the following two natural hypotheses for the instances of \pSC{}: (a)~$S_i \neq S_j$, $i,j = 1, \ldots, m$ and (b)~$S_i \nsubseteq S_j$, $i,j = 1, \ldots, m$.

\pSC{} is a well-known problem met in many real-world applications. 
To the best of our knowledge, it has been studied for the first time in the late seventies by Cornuejols et al.~\cite{cfn}.
This combinatorial problem originated from a financial application, where one wishes to find an optimal location of bank accounts in order to maximize clearing time. 
Since then, it is used for modeling real problems met in several areas such as databases, social networks, sensor placement, information retrieval, etc. 
A non-exhaustive list of references to such applications can be found in Badanidiyuru and al.~\cite{DBLP:conf/compgeom/BadanidiyuruKL12}.

\pVC{}, the graph version of \pSC{} is defined as follows: given a graph $G=(V,E)$ and two integers~$k$ and~$p$, one wants to determine~$k$ vertices that cover at least~$p$ edges. 
\pVC is a special case of \pSC where any element of~$X$ belongs to exactly two sets of~$\mathcal{S}$.

Both \pSC{} and \pVC{} are very important problems, since they are natural generalizations of \SC{} and \VC{}, respectively. Both are NP-hard  (setting $p=n$, \pSC{} becomes the seminal \SC{} problem; setting $p = |E|$, \pVC{} coincides with the \VC{} problem).

\pSC{} is known to be approximable within a factor $1-\nicefrac{1}{e}$ (to our knowledge, it is the only polynomial approximation result known for \pSC on general instances) but, for any $\epsilon > 0$, no polynomial algorithm can approximate it within ratio $1-\nicefrac{1}{e} + \epsilon$  unless $\mathrm{P} = \mathrm{NP}$~\cite{feigescj}, while \SC{} is polynomially inapproximable within ratio~$(1-\epsilon)\ln{n}$, for arbitrarily small $\epsilon > 0$, unless $\mathrm{P} = \mathrm{NP}$~\cite{DBLP:conf/approx/Moshkovitz12}.  On the other hand, \pVC{}, is APX-hard and the best-known approximation ratio for this problem, in general graphs, is bounded below by~$\nicefrac{3}{4}$, obtained by a very smart linear programming method by Ageev et al.~\cite{ageev}.

\pSAT is a satisfiability problem closely related to \pSC{}.
It is also a natural ``fixed cardinality generalization'' of \textsc{max sat}. 
In \pSAT{}, we are given a CNF on~$n$ variables and~$m$ clauses and we ask for setting to \textit{true} at most~$k$ variables satisfying at least~$p$ clauses. 
One may observe that \pSAT{} without negation is \pSC{}.

The goal of the paper is to establish several parameterized results for \pSC{} and for \pSAT{}. 
We mainly study the parameterized approximability of the former and exact parameterization of the latter. 

\section{Preliminaries}

We first give the basic definitions of the parameterized complexity theory. 
A parameterized problem~$(\Pi,k)$ is said \textit{fixed-parameter tractable} (or in the class~FPT) with respect to a parameter~$k$ if it can be solved by an algorithm with running time $f(k)\cdot|I|^{O(1)}$ time (in \textit{fpt-time}), where~$f$ is some computable function and~$|I|$ is the instance size. 
Such algorithms are called \textit{fixed-parameter tractable algorithms}, or FPT algorithms.
A parameterized reduction (or FPT reduction) from a problem $\Pi_1$ to a problem $\Pi_2$ is a mapping of an instance~$(I,k)$ of~$\Pi_1$ to an instance~$(I',k')$ of~$\Pi_2$, computable in time $f(k)\cdot |I|^{O(1)}$, such that $(I,k) \in \Pi_1 \Leftrightarrow (I',k')\in \Pi_2$, $k' \le g(k)$, and $|I'| \leqslant h(k)\cdot |I|^{O(1)}$ for some computable functions~$f$, $g$, and $h$.
This seemingly technical definition is just tailored to ensure that if $\Pi_2$ is in FPT and there is an FPT reduction from $\Pi_1$ to $\Pi_2$, then $\Pi_1$ is also in FPT.  
Some problems such as \textsc{Clique} parameterized by the solution size are not in FPT.
In fact, there is a whole hierarchy of classes beyond FPT: FPT $\subseteq$ W[1] $\subseteq$ W[2] $\subseteq \dots \subseteq$ W[P] $\subseteq$ XP. 
It is commonly believed that FPT $\neq$ W[1].

We need some additional definitions in order to give a precise meaning to those classes. 
A \textit{boolean circuit} is a directed acyclic graph where every vertex of in-degree~0 is an \textit{input vertex}, every vertex of in-degree~1 is a \textit{negation vertex} and every vertex of in-degree greater than~2 is either an \textit{and-vertex} or an \textit{or-vertex}. 
Exactly one vertex with out-degree 0 is the \textit{output vertex}. 
The \textit{depth} of such a circuit is the maximum length of a path from an input vertex to the output vertex, and the \textit{weft} of such circuit is the maximum number of \emph{large} vertices on a path from an input vertex to the output vertex.
A vertex is \emph{large} if its in-degree exceeds some pre-agreed constant bound. 
Giving boolean values to the input vertices determines the value of every vertex in the classic way, and in particular, if the output vertex receive value true for a given assignment, we say that this assignment satisfies the circuit.

The \textsc{weighted circuit satisfiability}~(WCS) problem takes as input a boolean circuit and an integer $k$ and decides if there is a satisfying assignment for this circuit with exactly $k$ input vertices set to true.
A parameterized problem~$\Pi$ belongs to the class~W[$t$], $t \geq 1$, if there is an FPT reduction from~$\Pi$ to~WCS restricted to circuits of weft at most~$t$.
A parameterized problem~$\Pi$ is hard for the class~W[$t$] (with $t \geq 1$) if, for any problem $\Pi'$ in W[$t$], there is an FPT reduction from $\Pi'$ to~$\Pi$; or equivalently if there is an FPT reduction from WCS restricted to circuits of weft at most~$t$ to $\Pi$.
A parameterized problem~$\Pi$ is W[$t$]-complete if it is W[$t$]-hard and in~W[$t$].
For example, \textsc{max independent set} parameterized by the size of the solution is W[1]-complete and \textsc{min dominating set} parameterized by the size of the solution is W[2]-complete.
The class~W[P] contains problems reducible to~WCS without constraints on the weft of the circuits.
The class~XP contains problems solvable in time~$|I|^{f(k)}$, where~$f$ is any computable function. 
See for example the monograph of Downey and Fellows for more details about fixed-parameter tractability~\cite{dowfel}.

We now go back to our problems. 
\pSC{} is W[2]-hard for the parameter~$k$ since, by setting $p=n$, we obtain an instance of \SC{} which is W[2]-hard. 
An FPT algorithm with respect to the standard parameter~$p$ is given by Bl\"{a}ser~\cite{blaserpartialcover}. 
Let us note that recently and independently certain aspects of parameterized complexity of \pSC have also been studied by Skowron and Faliszewski~\cite{AAAI159674}. 
These results are summarized in Table~\ref{table:complSC}.



Note that different parameterizations are possible for the same problem. 
In fact, one can also parameterize a problem by a combination of parameters (instead of just one parameter).
A \emph{multiparameterization} by parameters $k_1, k_2, \ldots, k_h$ consists of taking $k_1+k_2+ \ldots+ k_h$ as parameter.   
Multiparameterization poses novel and interesting open questions. 
For \pSC{}, for instance, several natural parameters as~$p$ (commonly called the ``standard parameter''),~$k$, $\Delta=\max_i\{|S_i|\}$ and $f=\max_i|\{j | x_i \in S_j\}|$ (commonly called the maximum frequency), can be jointly involved in a complexity study of the problem. 

We first give multiparameterization of \pSC{} with respect parameters~$k$ and~$p$ (Section~\ref{psc}). 
The most important part of this section is Subsection~\ref{approxissues} dedicated to the study of the parameterized approximation of \pSC{} for these two parameters. Consider a problem~$\Pi$, parameterized by some parameter~$\pi$. Then, we say that~$\Pi$ is {\em parameterized $r$-approximable}  if there exists an algorithm~\texttt{A} that is FPT when parameterized by~$\pi$ such that:
\begin{itemize}
\item if~$\Pi$ is a minimization problem, then for any instance~$I$ of~$\Pi$ where $\pi \leqslant \beta$,~\textsc{A} produces a solution with value at most~$r\beta$; otherwise it returns any solution (which can be smaller or greater than~$r\beta$);
\item if~$\Pi$ is a maximization problem, then for any instance~$I$ of~$\Pi$ where $\pi \geqslant \beta$,~\textsc{A} produces a solution with value at least~$r\beta$; otherwise it returns any solution (which can be smaller or greater than~$r\beta$).
\end{itemize} 
This line of research was initiated by three independent works~\cite{dofemcciwpec,caihuiwpec,chgrogruiwpec}. For an excellent overview, see the survey of Marx~\cite{marx-approx}. 
It aims at beating polynomial approximation barriers by offering more generous running time. 
The underlying question motivating Subsection~\ref{approxissues} is \emph{to what extent parameterized approximation is able to do it for \pSC{}?}.

Skowron and Faliszewski show, in~\cite{AAAI159674}, that it is possible when the parameter considered is $k+f$, where~$f$ is the frequency of the \pSC-instance, i.e., the maximum number of sets in~$\mathcal{S}$, a ground element belongs to. 
But what happens when considering only~$k$ instead?

For parameter~$k$ we mainly show a conditional result, informally, a parameterized (with respect to~$k$) approximation of \pSC within ratio greater than $1 - \nicefrac{1}{e} + \epsilon$, for some $\epsilon >0$, would lead to a parameterized approximation of \SC (with respect to the standard parameter) within ratio $(1-\varepsilon)\ln n$, for some fixed $\varepsilon > 0$. Even if this is a conditional result, we conjecture that the right answer is negative, i.e., that \pSC{} is inapproximable within ratio greater than $1 - \nicefrac{1}{e} + \epsilon$, for any $\epsilon >0$, in FPT time parameterized by~$k$. We also give in Subsection~\ref{approxissues} a weaker negative result for \pSC{}, namely, that under the same parameterization, it is inapproximable within ratio $1 - (\nicefrac{1}{n})^{\sqrt[4]{\ln{n}}}$, unless ETH\footnote{\textit{Exponential Time Hypothesis}: 
there is a real number $\delta>0$ such that \textsc{3-sat} is not solvable in $O(2^{\delta n})$ on instances with $n$ variables.} fails. 

Let us mention that \pVC{} can be approximately solved within ratio $1 - \epsilon$, for any fixed $\epsilon > 0$, in FPT time parameterized by~$k$~\cite{marx-approx}.

For parameter~$p$, we show that \pSC can be solved within ratio (strictly) greater than $1 - \nicefrac{1}{e}$ in time FPT parameterized by~$p$, which is smaller than that needed for the exact solution of the problem. 
 

In Section~\ref{satsec}, we settle the parameterized complexity of \pSAT{}, where, given a CNF on~$n$ variables and~$m$ clauses, one asks for setting to \textit{true} at most~$k$ variables satisfying at least~$p$ clauses.
The main result is that \pSAT{} is FPT with respect to parameter $p$.

To prove that, we refine a technique for obtaining multiparameterized FPT algorithms developed by Bonnet et al.~\cite{multiparameteralgorithmica}, called \emph{greediness-for-parameterization}, which is based on \emph{branching} algorithms.
Roughly, a branching algorithm extends a \emph{partial} solution at each recursion step.
The execution of such an algorithm can be seen as a \emph{branching tree}.
The best among the \emph{complete} solutions at the leaves of the branching tree, is output. 
The basic idea of the technique is to branch on:
\begin{itemize}
\item a greedy extension of the partial solution;
\item other extensions in the neighborhood of the greedy extension.
\end{itemize}
The soundness of the algorithm lies on the fact that if none of the above extensions of the partial solution is done by a supposed optimal solution, then the greedy choice stays optimal at the end.
Although the techniques are not the same, greediness-for-parameterization shares some common points with the \textit{greedy localization} technique (see~\cite{chenetalgreedylocalizationfsttcs01,dehneetalliwpec04,liuetaliwpec06} for some applications). Here, one uses a local search approach: one starts from a computed approximate solution and turns it to an optimal solution. 
However, greedy localization technique is less general than greediness-for-parameterization, since it suits maximization problems only.  

In Section~\ref{weft}, we suggest an enhanced weft hierarchy called \emph{counting weft hierarchy} dedicated to those cardinality-constrained problems, such as \pSAT{} and \pSC{} which are W[$i$]-hard for some $i$, and in~W[P], but not even known to be in~W[$j$] for some integer~$j$.
Related issues have been discussed by Fellows et al.~\cite{fellows10}. 


\section{Parameterizations for \pSC{}}\label{psc}

\subsection{Exact Parameterization}\label{exact}

As mentioned in Section~\ref{intro}, in the \pSC{} problem, we are given a family of subsets $\mathcal S=\{S_1,\ldots,S_m\}$ over a set of elements $X=\{x_1,\ldots,x_n\}$, and two integers~$p$ and~$k$.
The goal is to find a subcollection of $\mathcal S$ of size $k$ that covers at least $p$ elements of $X$.

Let us first note that, as $p \leqslant \Delta k$, the FPT result for \pSC{} with respect to parameter~$p$ presented by Bl\"{a}ser~\cite{blaserpartialcover}, immediately implies that this problem is also~FPT when parameterized by $k + \Delta$. 
An alternative proof using greediness-for-parameterization is given by Bonnet et al.~\cite{DBLP:journals/corr/BonnetPS13}. It might be worth reading it since it is a good introduction to the FPT algorithm for \pSAT{} in Section~\ref{satsec}.

We now explain why \pSC{} parameterized by $k+f$ is W[1]-hard. 
Each instance~$(\mathcal{S},X)$ of \pSC{} such that $f=2$ (that is, each element appears in at most two sets) can be seen as a graph whose vertices are the sets in~$\mathcal{S}$, and where there is an edge between two vertices if the corresponding sets share at least one element. 
Therefore \pSC{} with frequency~2 is equivalent to the \pVC{} problem where, given a graph~$G$ and a number~$k$, the goal is to cover at least~$p$ edges with~$k$ vertices. 
Thus, \pVC, W[1]-hard with respect to~$k$~\cite{cai}, is a restricted case of \pSC{}.

Note that in the reduction above the maximum set-cardinality~$\Delta$ in an instance of \pSC{} with $f=2$, coincides with the maximum degree of the derived graph. Hence, with the same argument, it can be shown that \pSC{} is not in XP when parameterized by~$\Delta+f$, since \pVC{} is NP-hard even in graphs with bounded degree (being, as mentioned above, a generalization of \VC that remains NP-hard even in these graphs).

Finally, in the following proposition, we prove that \pSC{} parameterized by~$k$ belongs to~W[P].
\begin{proposition}\label{prop3}
\pSC{} parameterized by~$k$ belongs to~W[P].
\end{proposition}
\begin{proof}
The proof is in exactly the same spirit with the proof by Cesati~\cite{Cesati03}. We reduce \pSC{} to \BNTMC{} which is a known W[P]-complete problem~\cite{dowfel} and defined as follows. Given a non-deterministic Turing machine~$M$, an input word~$w$, an integer~$n$ encoded in unary and a positive integer~$k$, does~$M(w)$ non-deterministically accept in at most~$n$ steps and using at most~$k$ deterministic steps?

Let $\mathcal I = (\mathcal S=\{S_1,\ldots,S_m\},p)$ be an instance of \pSC{}. Build a Turing Machine~$M$ with three tapes~$T_1$, $T_2$ and~$T_3$. Tape~$T_1$ is dedicated to non-deterministic guess. Write there the~$k$ sets $S_{a_1},\ldots,S_{a_k}$. Then, the head of~$T_1$ runs through all the elements and when a new element is found it is written down on the second tape. The third tape counts the number of already covered elements. If this number reaches~$p$, then~$M$ accepts. Thus, there exist~$k$ non-deterministic steps, and a polynomial (in $|\mathcal I|$) number of deterministic steps (precisely,~$O(|\mathcal I|^2)$).~\qed
\end{proof}

The results mentioned in this paragraph as well as literature results are summarized in Table~\ref{table:complSC}.

\begin{table}
\begin{center}
\begin{tabular}{c|c|c|c|c|c|c}
\textbf{Parameter:} 	& $\Delta + f$	&  $k$ 		&  $k + f$	&   $k+\Delta$ or $p$ & $n-p$ & $(n-p)+k$  \\\hline 
\textbf{Status:}		& $\notin$ XP & \begin{tabular}{c} W[2]-hard\\ in  W[P] (Prop.~\ref{prop3})\end{tabular}  &  \begin{tabular}{c} W[1]-hard\\ in  W[P]\end{tabular}	& FPT~\cite{blaserpartialcover,DBLP:journals/corr/BonnetPS13} & $\notin$ XP~\cite{AAAI159674} & W[2]-complete~\cite{AAAI159674}
\end{tabular} 
\end{center}
\caption{Exact parameterized complexity of \pSC{} for different parameters.}\label{table:complSC}
\end{table}

\subsection{Approximation Issues}\label{approxissues}

Let us now handle parameterized approximation of \pSC{}. We first prove the following basic lemma that is an easy generalization of Proposition~5.2 given by Feige~\cite{feigescj}.
\begin{lemma}\label{feigegenlem}
Any $r$-approximation algorithm, parameterized by~$k$, for \pSC{} can be transformed into an FPT $t$-approximation algorithm, parameterized by the standard parameter (optimum value), for \SC{} where:
 $$
 t < \left\lceil\frac{-\ln{n}}{\ln(1-r)}\right\rceil
 $$
 \end{lemma}
\begin{proof}
The basic idea is similar to the idea of Feige~\cite{feigescj} (Proposition~5.2). Its key ingredient is the following. Consider some algorithm \kscalg{} that solves \pSC{}. Then, it can iteratively be used to solve \SC{} as follows. Consider an instance $I = (\mathcal{F},U)$ of \SC{} where~$\mathcal{F}$ is a family of subsets of a ground set~$U$. 
Iteratively run \kscalg{} for $k = 1, \ldots, m$ (where~$m$ is the size of~$\mathcal{F}$). 
Eventually, one value of $k$ will equal the value of the optimal solution for \SC{}. 
Let us reason with respect to this value of~$k$, denoted by~$k_0$. Furthermore, assume that \kscalg{} achieves approximation ratio~$r$ for \pSC{}. 
Invoke it with value~$k_0$, (note that now $p=n$, the size of the ground set~$U$), remove the ground elements covered, store the~$k_0$ elements used and relaunch it with value~$k_0$, until all ground elements of~$U$ are removed. 
Since it is assumed to achieve approximation ratio~$r$  after its $\ell$-th execution at most $(1-r)^{\ell}n$ ground elements remain uncovered. 
Finally, suppose that after~$t$ executions, all ground elements are removed (covered). Then, the~$t k_0$ subsets stored form a $t$-approximate solution for the \SC{}-instance, where~$t$ satisfies (after some very simple algebra):
\begin{equation}\label{scbasicratio}
(1-r)^{t}n < 1 \Longrightarrow t < \left\lceil\frac{-\ln{n}}{\ln(1-r)}\right\rceil
\end{equation}
Moreover, observe that the complexity of the algorithm derived for \SC{} is at most~$m$ times the complexity of \kscalg; so, it remains FPT with respect to the optimal value for \SC{}  .~\qed
\end{proof}
Recall that, as mentioned in the beginning of Section~\ref{intro}, \pSC{} is inapproximable in polynomial time within ratio $1- \nicefrac{1}{e} + \epsilon$, for any $\epsilon >0$, unless $\mathrm{P} = \mathrm{NP}$~\cite{feigescj}. 
We first prove in the sequel a conditional result, informally, getting such a ratio even in FPT time parameterized by~$k$, is a rather difficult task. More precisely, we prove that if this were possible, then we could get, in parameterized time, an approximation ratio of~$(1-\epsilon)\ln(\nicefrac{n}{\ln n})$ for \SC{}, for some fixed $\epsilon > 0$. Even if this result is, properly speaking, a conditional result, it gives, in some sense, the measure of the difficulty of approximating \pSC{} within ratio strictly better than $1- \nicefrac{1}{e}$ in FPT time parameterized by~$k$. 
Next, we  prove an FPT inapproximability result, namely that FPT approximation of \pSC{} within ratio greater than $1 - (\nicefrac{1}{n})^{\sqrt[4]{\ln{k}}}$
is impossible unless $\text{W[2] = FPT}$. 
\begin{proposition}\label{inapproxpsccond}
\pSC{} parameterized by~$k$ is inapproximable within ratio $(1- \nicefrac{1}{e} + \epsilon)$, for any $\epsilon \in [0, \nicefrac{1}{e})$, unless \SC{} is approximable within ratio~$(1-\eta)\ln(\nicefrac{n}{\ln{n}})$, for some fixed $\eta> 0$, in FPT time parameterized by the value of the optimum.
\end{proposition}
\begin{proof}
Revisit Lemma~\ref{feigegenlem}, take $r = 1 - (\nicefrac{1}{e}) + \epsilon$, for some $\epsilon \in [0, \nicefrac{1}{e})$ and assume that the \kscalg{} of Lemma~\ref{feigegenlem} (that is FPT in~$k$) achieves approximation ratio~$r$. Then, in order to prove the result claimed, follow the procedure described in Lemma~\ref{feigegenlem} until there are at most~$\ln n$, say~$c\ln n$ for some $c \leqslant 1$, uncovered elements in~$U$ and solve the remaining instance by, say, the best known exact algorithm which works within~$O^*(2^n)$ in instances with ground set-size~$n$~\cite{incexcsetpart}. Since the surviving ground set has size~$c\ln n$, it is polynomial to optimally solve it. Reasoning exactly as in Lemma~\ref{feigegenlem} we get:
\begin{eqnarray*}
n(1-r)^t = c\ln{n} \Longrightarrow t &=& \frac{\ln{c} + \ln\ln{n} - \ln{n}}{\ln(1-r)} \leq \frac{\ln\ln{n} - \ln{n}}{\ln(1-r)} \simeq \frac{\ln\ln{n} - \ln{n}}{-\epsilon e - 1} \\ 
&=& \frac{\ln{n} - \ln\ln{n}}{1+\epsilon e} = \frac{1}{1+\epsilon e}\ln\left(\frac{n}{\ln{n}}\right) 
\end{eqnarray*}
Setting $\eta = \nicefrac{\epsilon e}{(1+\epsilon e)}$,
%
%
%
the proof of the proposition is concluded.~\qed
\end{proof}
As one can see, the result of Proposition~\ref{inapproxpsccond} is conditional and relates the parameterized approximability of \pSC{} within ratios better than the one achieved in polynomial time to the parameterized approximability of \SC{} within ratios that are almost the same (in fact slightly smaller) as the one polynomially achieved for this problem. Furthermore this ratio is tight for the polynomial time (recall that it is {NP}-hard to approximate \SC{} within ratio $(1-\varepsilon)\ln{n}$, for arbitrarily small $\varepsilon > 0$~\cite{DBLP:conf/approx/Moshkovitz12}). We conjecture that the real parameterized (with respect to the optimum) inapproximability bound of \SC{} is~$O(\log n)$, so that the inapproximability bound (conditionally) conjectured by Proposition~\ref{inapproxpsccond} is the correct one. 
But, unfortunately, we have not been able to prove it until now and the negative result by Moshkovitz in~\cite{DBLP:conf/approx/Moshkovitz12} does not seem to be usable as it is for the parameterized inapproximability of \pSC{}.

In what follows, in the spirit of Lemma~\ref{feigegenlem} and of Proposition~\ref{inapproxpsccond} and based upon a recent result by Chen and Lin~\cite{DBLP:journals/corr/ChenL15b}, we show a weaker upper bound for the parameterized approximability of \pSC{} with respect to~$k$.

Consider the \DS{} problem defined as follows:\textit{ given a graph $G= (V,E)$, determine a minimum size vertex subset $D \subseteq V$ such that every vertex of~$V$ is either in~$D$ or has a neighbor in~$D$}. There exists a well-known approximability preserving reduction from \DS{} to \SC{} that works as follows: given a graph~$G = (V,E)$ of order~$n$, instance of \DS{}, we transform it into an instance $I = (\mathcal{F},U)$ of \SC{} as follows:
\begin{itemize}
\item $\mathcal{F} = \{F_1, \ldots, F_n\}$;
\item $U = \{u_1, \ldots, u_n\}$;
\item $\forall i \in \{1, \ldots, n\}$, $F_i = \{u_j: v_j \in \Gamma[v_i]\}$, where~$\Gamma[v_i]$ denotes the closed neighborhood of vertex $v_i \in V$. 
\end{itemize}
Then, it is easy to see that any dominating set~$D$ of~$G$ corresponds to a set cover~$\mathcal{F}'$ of~$I$ of the same cardinality by simply considering in~$\mathcal{F}'$ the subsets of~$\mathcal{F}$ having the same indices with the vertices of~$D$ and vice-versa. An immediate consequence of this reduction is that both problems share the same approximation ratios and inapproximability bounds.

Recently, Chen and Lin~\cite{DBLP:journals/corr/ChenL15b} have proved that, \textit{under~ETH, no~FPT algorithm for~\DS{} can achieve approximation ratio smaller than, or equal to,~$\sqrt[4+\epsilon]{\ln\gamma(G)}$, for any positive constant~$\epsilon$, where~$\gamma(G)$ denotes the cardinality of a minimum dominating set in~$G$}. 
The reduction just described, immediately transfers this lower bound to \SC{}; so the following holds for this latter problem: \textit{under~ETH, no~FPT algorithm for~\SC{} can achieve approximation ratio smaller than, or equal to,~$\sqrt[4+\epsilon]{\ln k_0}$, for any positive constant~$\epsilon$, where~$k_0$ denotes the cardinality of a minimum set cover instance~$(\mathcal{F},U)$}. 

Based upon the result by Chen and Lin~\cite{DBLP:journals/corr/ChenL15b}, Lemma~\ref{feigegenlem}, and the approximation-preserving reduction from \DS{} to \pSC{} given just above, the following can be proved.
\begin{proposition}\label{kscparaminapprox}
\pSC{} is inapproximable within ratio $1 - (\nicefrac{1}{n})^{\sqrt[4]{\ln{k}}}$ in FPT time parameterized by~$k$, unless ETH fails.
\end{proposition}
\begin{proof}
The proof follows the one of Lemma~\ref{feigegenlem}. From~(\ref{scbasicratio}), taking for~$t$ (the ratio of \SC{}) the $\sqrt[4]{\ln{k_0}}$-inapproximability bound derived by the the reduction above and by Chen and Lin in~\cite{DBLP:journals/corr/ChenL15b} and omitting (in order to simplify calculations) the ceiling in~(\ref{scbasicratio}), some elementary algebra leads to:
$$
\ln(1-r) \leqslant \frac{-\ln{n}}{t} \leqslant \frac{-\ln{n}}{\sqrt[4]{\ln{k_0}}} \Longrightarrow 1 - r \geqslant \left(\frac{1}{n}\right)^{\sqrt[4]{\ln{K_0}}} \Longrightarrow r \leqslant 1 - \left(\frac{1}{n}\right)^{\sqrt[4]{\ln{k_0}}}
$$
as claimed.~\qed
\end{proof}
Let us note that a parameterized inapproximability bound weaker than that of Proposition~\ref{kscparaminapprox} can be obtained as follows. Consider an instance~$G$ of \DS{} and transform it to an instance $I = (\mathcal{F},U)$ of \SC{} by the transformation seen above. Assume now that \kscalg{} achieves ratio $1 - \nicefrac{c}{n}$ for some fixed $c > 1$. Then, just run \kscalg{} only once for every~$k$. Assuming that \kscalg{} runs in time~$O(p(n)F(k))$ for some polynomial~$p$, the whole of runs will take~$m\cdot O(p(n)F(k))$-time that remains FPT in~$k$. For~$k_0$ (the value of the optimal solution for~\SC{} in instance~$I$), it holds that, after this run, at most $n - n(1 - (\nicefrac{c}{n})) = c$ elements will remain uncovered. Any (non-trivial) cover for them uses at most~$c$ sets to cover them. In this case, the procedure above achieves an \emph{additive} approximation error $c+1$ (recall that~$c$ is fixed) for \SC{}, and this ratio is identically transferred to \DS{} via the reduction. But for \DS{}, achievement of any constant additive approximation error is W[2]-hard~\cite{Downey:appDS}. So, the following corollary holds.
\begin{corollary}
\pSC{} is inapproximable within ratio $1 - \nicefrac{c}{n}$ in time parameterized by~$k$, for any constant $c > 1$, unless $\text{W[2] = FPT}$.
\end{corollary}
%

For the rest of the section, we relax the optimality requirement for the \pSC-solution and we show that we can devise an approximation algorithm with ratio strictly better than $1 - \nicefrac{1}{e}$ (beating so the polynomial inapproximability bound of Feige~\cite{feigescj}), that runs in FPT time parameterized by~$k$ and~$\Delta$ but whose complexity is lower (although depending on the accuracy) than the best exact parameterized complexity for \pSC. 

In fact, we are going to prove a stronger result, claiming that, for any polynomial approximation ratio~$r$ achieved by some polynomial time algorithm \texttt{APPROX} and any parameterized algorithm~$\mathtt{PARAM}(\pi)$ for \pSC{}, where~$\pi$ can be a single parameter or a vector of parameters,~$r$ can be improved in FPT time parameterized by~$\pi$ by an algorithm whose running time is smaller than the one of~$\mathtt{PARAM}(\pi)$.

The way of doing it is to built a kind of hybrid algorithm that, informally, in the case of \pSC{} we deal with, where $\pi = (k,\Delta)$, it works as follows:
\begin{itemize}
\item it solves \pSC{} by invoking $\mathtt{PARAM}(\pi')$, with $\pi' = (k', \Delta)$ for some $k' < k$, stores the solution computed, and removes it from the initial instance; 
\item it invokes \texttt{APPROX} on the surviving instance and stores the solution obtained;
\item it takes the union of the two solutions.
\end{itemize}
Here, our objective is to establish a trade-off between the solution quality and the running time. 
Therefore, the running time of the presented algorithms depend on the approximation ratio.

Assume an FPT (exact) algorithm running in time~$O^*(F(k,\Delta))$ for some function~$F$ (that is Algorithm $\mathtt{PARAM}(k,\Delta)$) together with an approximation algorithm \texttt{APPROX} achieving approximation ratio~$r$ for \pSC and consider Algorithm~\ref{pscschema}, called \pscschema{} in what follows, running on an instance~$(\mathcal{S},X)$ of \pSC{}, where~$X$ is the ground set and~$\mathcal S$ a family of subsets of~$X$. 

\begin{algorithm}
\KwIn{An instance ~$(\mathcal{S},X,k)$ of \pSC{}.}
\KwOut{A subfamily $\hat{\mathcal{T}} \subset \mathcal S$ containing~$k$ subsets.}
fix some $\varepsilon > 0$\;
take $\mu \in (0,1)$ such that $\varepsilon > (1-e^{-1})\mu^2 - (1 -2e^{-1})\mu$\;
set $k' = \mu k$\; 
run Algorithm \texttt{PARAM}$(k',\Delta)$ and store the solution computed (denoted by~$\mathcal{T}_1$); let~$X_1$ be the subset of~$X$ covered by~$\mathcal{T}_1$\;
set $\mathcal{S}' = \mathcal{S} \setminus \mathcal{T}_1$, $X' = X \setminus X_1$ and $k'' = k - k' = (1-\mu) k$\; 
run the $r$-approximation algorithm \texttt{APPROX}
on the \pSC{}-instance~$(\mathcal{S}',X',k'')$ and store the solution~$\mathcal{T}_2$ computed\;
\Return{$\hat{\mathcal{T}} = \mathcal{T}_1 \cup \mathcal{T}_2$}\;
\caption{A description of the algorithm \pscschema.}
\label{pscschema}
\end{algorithm}
Solution~$\hat{\mathcal{T}} $ computed by \pscschema{} has cardinality~$k$, i.e., it is feasible for \pSC{}. 
Let us now analyze it in the following proposition.

\begin{proposition}\label{paramapproxp}
Given a polynomial time approximation algorithm with ratio $r$ for \pSC{},  for any $\mu \in (2- (\nicefrac{1}{r}),1]$, \pSC{} can be approximated within ratio $r(1-\mu)^2 + \mu > r$, in~$O^*(F(\mu k,\Delta))$-time, where $\Delta$ is the maximum cardinality of a set.
\end{proposition}
\begin{proof}
Fix an optimal solution~$\mathcal{T}^*$ and denote by~$X^*$ the subset of~$X$ covered by~$\mathcal{T}^*$. 
Recall that in Algorithm~\ref{pscschema}, we denote by $X_1$ the subset of $X$ covered by the set $\mathcal{T}_1$ computed by \texttt{PARAM}$(k',\Delta)$.
Obviously: 
\begin{equation}\label{paramp1}
\left|X_1\right| \geqslant \mu\left|X^*\right|
\end{equation}
Fix an optimal \textsc{max} $k''$-\textsc{set cover}-solution~$\bar{\mathcal{T}}$ of~$(\mathcal{S}',X')$, denote by~$\bar{X}^*$ the set of elements of~$X$ covered by~$\bar{\mathcal{T}}$ and set $\tilde{X} = X^* \cap X_1$. Remark now the following facts:
\begin{enumerate}
\item\label{fact1} $X^* \setminus \tilde{X}$ is covered with at least~$k''$ sets in~$\mathcal{T}^*$ (denote by~$\mathcal{T}''^{*}$ this system); otherwise, the sets of~$\mathcal{T}^*$ covering $X^* \setminus \tilde{X}$ together with~$\mathcal{T}_1$ would be a solution better than~$\mathcal{T}^*$; indeed, if the elements of $X^* \setminus \tilde{X}$ were covered with less than~$k''$ sets, i.e., if $|\mathcal{T}''^{*}| < k''$, then, since $\tilde{X} \subseteq X_1$, $\mathcal{T}_1 \cup \mathcal{T}''^{*}$ would be a set system of size $|\mathcal{T}_1 \cup \mathcal{T}''^{*}| < k$ covering~$|X^*|$ elements; in this case, completing (even greedily) the family $\mathcal{T}_1 \cup \mathcal{T}''^{*}$ with $k - |\mathcal{T}_1 \cup \mathcal{T}''^{*}|$ sets would lead to a $k$-sets subfamily of~$\mathcal{S}$ covering more than~$|X^*|$ ground elements, absurd since~$X^*$ is the value of an optimal solution for \pSC; 
\item\label{fact3} the elements of $X^* \setminus \tilde{X}$ are still present in the instance~$(\mathcal{S}',X')$ where the $r$-approximation algorithm \texttt{APPROX} is invoked, as well as the subsets of~$\mathcal{S}$ covering them, i.e., the sets of~$\mathcal{T}''^{*}$;
\item\label{fact2} hence, the~$k''$ ``best''\footnote{In the sense that they cover the most of the elements covered by any other union of~$k''$ sets of~$\mathcal{T}''^{*}$.}  sets of~$\mathcal{T}''^{*}$ form a feasible solution for \textsc{max} $k''$-\textsc{set cover} in~$(\mathcal{S}',X')$ and cover more than $(\nicefrac{k''}{k})|X^* \setminus \tilde{X}|$ elements of~$X^* \setminus \tilde{X}$.
\end{enumerate}
Combining Facts~\ref{fact1}, \ref{fact3} and~\ref{fact2} and taking into account that~$\mathcal{T}_2$ is an $r$-approximation for \pSC, the following holds denoting by~$X_2$ the subset of~$X$ covered by~$\mathcal{T}_2$:
\begin{eqnarray}
\left|X_2\right| &\geqslant& r\cdot \left|\bar{X}^*\right| \geqslant r\cdot\frac{k''}{k}\cdot\left(\left|X^* \setminus \tilde{X}\right|\right) 
=  r\cdot \frac{k-k'}{k}\cdot \left(\left|X^* \setminus \tilde{X}\right|\right) \nonumber \\
&=& r\cdot(1-\mu)\left(\left|X^* \setminus \tilde{X}\right|\right) \label{paramp2} \\
\left|X^* \setminus \tilde{X}\right| &=& \left|X^*\right| - \left|\tilde{X}\right| \geqslant 
\left|X^*\right| - \left|X_1\right| \label{paramp3}
\end{eqnarray}
Putting together~(\ref{paramp1}), (\ref{paramp2}) and~(\ref{paramp3}), we get the following for the approximation ratio of Algorithm \pscschema{}:
\begin{eqnarray}\label{finalratio}
\frac{\left|X_1\right| + \left|X_2\right|}{\left|X^*\right|} &\geqslant& \frac{\left|X_1\right| + r\cdot(1-\mu)\cdot\left(\left|X^*\right| - \left|X_1\right|\right)}{\left|X^*\right|}   \\
&\geqslant& \frac{r\cdot(1-\mu)\cdot\left|X^*\right| + \left[1-r\cdot(1-\mu)\right]\cdot\left|X_1\right|}{\left|X^*\right|} \nonumber \\
&\geqslant& \frac{\left|X^*\right|\cdot\left[r\cdot(1-\mu) + \mu\cdot\left[1-r\cdot(1-\mu)\right]\right]}{\left|X^*\right|} \nonumber \\
&=& r\cdot(1-\mu) + \mu\cdot\left[1-r\cdot(1-\mu)\right] \nonumber = r(1-\mu)^2 + \mu
\end{eqnarray}
Ratio in~(\ref{finalratio}) is at least~$r$, for any $\mu \in ((2- (\nicefrac{1}{r})), 1]$.

For the overall running time, it suffices to observe that, since \texttt{APPROX} runs in polynomial time, the running time of \pscschema{} is dominated by that of \texttt{PARAM} invoked within Algorithm~\ref{pscschema}. Thus, the whole complexity of Algorithm \pscschema{} becomes~$O^*(F(\mu k,\Delta))$, as claimed.~\qed
\end{proof}
Revisit now the proof of Proposition~\ref{paramapproxp} and set $r = 1 - e^{-1}$. Then,~(\ref{finalratio}) becomes:
\begin{eqnarray*}
\frac{\left|X_1\right| + \left|X_2\right|}{\left|X^*\right|} &\geqslant& \frac{\left|X_1\right| + \left(1-e^{-1}\right)\cdot(1-\mu)\cdot\left(\left|X^*\right| - \left|X_1\right|\right)}{\left|X^*\right|} \nonumber \\
&\geqslant& \frac{\left(1-e^{-1}\right)\cdot(1-\mu)\cdot\left|X^*\right| + \left[1-\left(1-e^{-1}\right)\cdot(1-\mu)\right]\cdot\left|X_1\right|}{\left|X^*\right|} \nonumber \\
&\geqslant& \frac{\left|X^*\right|\cdot\left[\left(1-e^{-1}\right)\cdot(1-\mu) + \mu\cdot\left[1-\left(1-e^{-1}\right)\cdot(1-\mu)\right]\right]}{\left|X^*\right|} \nonumber \\
&=& \left(1-e^{-1}\right)\cdot(1-\mu) + \mu\cdot\left[1-\left(1-e^{-1}\right)\cdot(1-\mu)\right] \nonumber \\
&=& \left(1-e^{-1}\right) -  \left(1-2e^{-1}\right)\cdot\mu + \left(1-e^{-1}\right)\cdot\mu^2
\end{eqnarray*}
This ratio is at least $1 - \nicefrac{1}{e} + \varepsilon$, for any $\varepsilon > (1-e^{-1})\mu^2 - (1-2e^{-1})\mu$, and the following corollary holds.
\begin{corollary}\label{paramapproxpcor}
For any $\mu > \nicefrac{(e-2)}{(e-1)}$ and any $\varepsilon > (1-e^{-1})\mu^2 - (1 -2e^{-1})\mu$, \pSC{} can be approximated within ratio $1 - \nicefrac{1}{e} + \varepsilon$, in~$O^*(F(\mu k,\Delta))$-time. 
\end{corollary}
For instance, let us take $\mu = 0.5 \geqslant \nicefrac{(e-2)}{(e-1)}$. Then application of Corollary~\ref{paramapproxpcor} leads, after some easy algebra to an approximation ratio at least $0.75 + (\nicefrac{0.25}{e}) \geqslant 0.841$ achieved with complexity~$F(\nicefrac{k}{2}, \Delta)$. Even if, in the case of \pSC{}, the potential of the running time seems to be quite small, we think that looking for this type of trade-offs between approximation ratios and running times in order to try to beat polynomial approximation barriers is an interesting research program.

In~\cite{maxkvcjoco}  an analogous result is given. There, the authors try to beat the APX-hardness of \textsc{max $k$-vertex cover}  by a moderately exponential approximation algorithm, i.e., by an approximation algorithm achieving ratio $1-\epsilon$, for any $\epsilon > 0$, and running in time that, although exponential, remains smaller than the running time of the best known exact algorithm for this problem. More precisely, the following is proved by Della Croce and Paschos~\cite{maxkvcjoco}. \textit{If~$T(n,k)$ is the running time of an exact algorithm and~$\rho$ the approximation ratio of some polynomial approximation algorithm for \pVC{}, then, for any $\epsilon > 0$, \pVC{} can be approximated within ratio $1 - \epsilon$ with worst-case running time~$T(n,[(2\rho-1)+\sqrt{1-4\epsilon\rho}]k/2\rho)$. }
 
\section{\pSAT{}}\label{satsec}

We now study a generalization of \pSC{}, the \pSAT{} problem.
%
Recall that in \pSAT{}, given a CNF formula on~$n$ variables and~$m$ clauses, the objective is to satisfy at least~$p$ clauses, by setting at most~$k$ variables to \textit{true}. 
\begin{proposition}\label{psathard}
\pSAT\ parameterized by $k$ is W[2]-hard and in W[P].
\end{proposition}
\begin{proof}
Setting $p = m$, \pSAT becomes \textsc{sat-$k$} that is W[2]-hard~\cite{niedermeier06} (under the name \textsc{weig\-h\-ted CNF-satisfiability}). Proof of membership of \pSAT{} in~W[P] can be done by an easy reduction of this problem to \BNTMC{}. One can guess within~$k$ non deterministic steps the variables to put to \textit{true} and then one can check in polynomial time whether, or not, at least~$p$ clauses are satisfied.~\qed
\end{proof}
Consider an instance~$\phi$ of \pSAT{} on a set of~$C$ clauses. 
Given any subset~$C'$ of~$C$, we denote by~occ$^+(x_i,C')$ the number of positive occurrences of the variable~$x_i$ in~$C'$, and by~occ$^-(x_i,C')$ the number of its negative occurrences. 
We set $f(x_i):=\mathrm{occ}^+(x_i,C)+\mathrm{occ}^-(x_i,C)$; so, the frequency of the formula is $f:=\max_i f(x_i)$.

Before proving that \pSAT{} is~FPT with respect to parameter~$p$, let us introduce some vocabulary on branching algorithms.
A \emph{partial} solution is a subset of a (\emph{complete}) solution.
A \emph{branching} algorithm is a recursive algorithm.
Its execution on an instance~$I$ can be seen as a tree, called \emph{branching tree}.
In this tree, each \emph{node} is labeled with a sub-instance of~$I$ together with a partial solution, or more generally with some data maintained by the algorithm.
The \emph{root} is labeled with~$I$ and a \emph{leaf} is a sub-instance that causes the branching algorithm to stop.
At a leaf, a complete solution is computed and returned.
When identifying a node~$v$ to its label (a sub-instance), the \emph{children} of a sub-instance are the subinstances which label the children of~$v$ in the branching tree.

We now prove this simple lemma.
\begin{lemma}\label{lem:pSAT}
\pSAT{} is solvable in time $O^*(2^m)$.
\end{lemma}
\begin{proof}
We take any variable $x$ that appears positively in at least one clause \emph{and} negatively in at least one clause.
We do the standard branching: either set $x$ to \textit{true} (and decrease $k$ by 1), or set~$x$ to \textit{false}. 
This branching satisfies at least one more clause in each branch.
Therefore, the branching tree is a subtree of the full binary tree with $2^m$ leaves. 
At a leaf $\ell$ of the branching tree, the remaining number of clauses is at most $m-\lambda$ where $\lambda$ is the depth of $\ell$.
The branching stops when each variable appears only negatively, or only positively.
At this point, the variables appearing only negatively can be set to \textit{false}.
This step is safe since we are constrained to put \emph{at most} (not exactly) $k$ variables to \textit{true}.
We end up with an instance containing only positive literals.
Therefore, at the leaves, the instances can now be seen as instances of \textsc{max $k$-hitting set} with at most $m-\lambda$ sets; or equivalently, \textsc{max $k$-set cover} with at most $m-\lambda$ elements which can be solved by standard dynamic programming in time $O^*(2^{m-\lambda})$.
So, the overall running time is~$O^*(2^m)$.~\qed
\end{proof}
We identify a solution of an instance of \pSAT to the set~$S$ of variables put to \emph{true}.
It induces a set~$C'$ of satisfied clauses.
The algorithm solving \pSAT we will present, performs two kinds of choices:~(1) setting a variable to \emph{true} and~(2) putting a clause which is not satisfied yet into a set~$C_s$ of clauses that should eventually be satisfied.
Putting a clause in~$C_s$ means that we commit to satisfy it later. 
At any node of the branching tree, a child corresponds to either performing choice~(1) for some given variable, or choice~(2) for some given clause.
A choice~(1) is \emph{in accordance} with~$S$ if it sets to \emph{true} a variable in~$S$.
A choice~(2) is \emph{in accordance} with~$S$ if it puts in~$C_S$ a clause satisfied by solution~$S$ (i.e., this clause is in~$C'$). 
A node~$v$ of the branching tree is \emph{in accordance} with~$S$ if all the choices made from the root to this node are in accordance with~$S$.
A node~$v$ of the branching tree \emph{deviates} from a solution~$S$ if it is in accordance with~$S$ but none of its children are in accordance with~$S$.

Let us give a toy example to clarify those notions. 
Assume a solution which sets~$x_2$ and~$x_3$ to \emph{true} and all the other variables to \emph{false}.
So, $S=\{x_2,x_3\}$.
And, this assignment satisfies the following set of clauses: $\{c_2,c_4,c_5,c_6,c_8,c_9\}$.
Say, the root of the branching tree has~$4$ children: setting~$x_1$ to \emph{true}, committing to satisfy~$c_1$, committing to satisfy~$c_2$, and committing to satisfy~$c_6$.
The two first children are not in accordance with~$S$, but the two last are.
Indeed,~$c_2$ and~$c_6$ are satisfied by~$S$.
Let us move to the child where we commit to satisfy~$c_2$.
Suppose this node has three children: setting~$x_2$ to \emph{true}, committing to satisfy~$c_1$, committing to satisfy~$c_3$.
We now move to the child where we set~$x_2$ to \emph{true}.
So far, we have only done choices in accordance with~$S$, so our current node is in accordance with~$S$.
Now, say, this new node~$v$ has three children: setting~$x_1$ to \emph{true}, committing to satisfy~$c_3$, committing to satisfy~$c_7$.
None of those choices is in accordance with~$S$, so~$v$ deviates from~$S$.
\begin{proposition}
\pSAT{} parameterized by~$p$ is FPT.
\end{proposition}
\begin{proof}
Let $(C,k,p)$ be an instance of \pSAT where $C$ is a set of clauses over a set of variables $X=\{x_1,x_2,\ldots,x_n\}$, $k$ is the maximum number of variables that can be set to true, and $p$ the minimum number of clauses to satisfy.
We can assume that $p < \nicefrac{m}{2}$. 
Indeed, if $p \geqslant \nicefrac{m}{2}$, the algorithm of Lemma~\ref{lem:pSAT} is an FPT algorithm. 
We also assume that the number of clauses containing only negative literals is bounded above by~$\nicefrac{m}{2}$.
Otherwise, setting all the variables to \textit{false}, satisfies more than~$p$ clauses.
We recall that we are not forced to set exactly~$k$ variables to \textit{true}, but \emph{at most}~$k$. 
We observe that instances such that $p<\nicefrac{f}{2}$ are always YES-instances, since one can set one variable~$x_i$ with frequency~$f$ to \textit{true} if $\mathrm{occ}^+(x_i,C) \geqslant \mathrm{occ}^-(x_i,C)$, and to \textit{false} otherwise, and set all the other variables to false.
This assignment does indeed satisfy at least $\max(\mathrm{occ}^+(x_i,C),\mathrm{occ}^-(x_i,C)) \geqslant \nicefrac{f}{2} < p$ clauses.

 Note also that instances such that $p<k$ are all YES-instances, too. 
Indeed, one can iteratively set to \textit{true}~$k$ variables such that at each step one satisfies at least one more clause. 
If, at some point this is no longer possible, then setting all the remaining variables to \textit{false} will satisfy all the clauses which do not initially contain only negative literals, that is at least half of the clauses, so more than~$p$ clauses. 
We may now assume that $p \geqslant \nicefrac{f}{2}$ and $p \geqslant k$, so our parameter might as well be $p+f+k$. 

We construct a branching algorithm which operates accordingly to a greedy criterion. 
A solution, or \emph{complete assignment}, is given by a set~$S$ of size up to~$k$ which contains all the variables set to \textit{true}. 
Additionally, we maintain a list~$C_s$ of clauses that we satisfy or commit to satisfy. 
For notational convenience we define $C_u:=C \setminus C_s$, $d_i(C'):=\mathrm{occ}^+(x_i,C')$, and let~$C^+(x_i,C')$ be the set of clauses in~$C'$ where~$x_i$ appears positively and~$C^-(x_i,C')$ the set of clauses where~$x_i$ appears negatively. 
Finally, $C(x_i,C') := C^+(x_i,C') \cup C^-(x_i,C')$.
Algorithm~\algpSATk (see Algorithm~\ref{alg-two}) is fairly simple.
We find the variable $x$ that, if set to \emph{true}, would satisfy the maximum number of clauses among the still unsatisfied clauses.
We branch on setting $x$ to \emph{true} (choice (1)) or for each still unsatisfied clause $c$ that $x$ would satisfy, on putting $c$ in $C_s$ (choice (2)).

\begin{algorithm}
\KwIn{A set $C$ of clauses on a set $X$ of variables.}
\KwOut{A subset $S \subseteq X$ of size at most $k$ such that setting all the variables in $S$ to \textit{true} and all the variables in $X \setminus S$ to \textit{false}, satisfies the greatest number of clauses in $C$.}
{set} $S=\emptyset$, $C_s=\emptyset$\;
\algone$(S,C_s)$:\\
\eIf{$|S|<k$ \emph{\textbf{and}} $|C_s|<p$}{
     pick the variable $x_i$ maximizing $d_i(C \setminus C_s)$\; 
     run \algone$(S \cup \{x_i\},C_s \cup C^+(x_i,C \setminus C_s))$\;
     \For{each clause $c \in C(x_i,C \setminus C_s)$}{
       run \algone$(S,C_s \cup \{c\})$\;}
}{\eIf{$|S|=k$}{ 
    store $S$\;
  }{($|C_s| \geqslant p$) store a complete assignment satisfying $C_s$, if possible\;
}}
\Return{the best among the solutions stored}\;
\caption{A description of the algorithm \algpSATk.}
\label{alg-two}
\end{algorithm}

The branching tree has depth at most $k+p$ and arity at most $f+1$, so the running time of \algpSATk is $O^*(2^p(f+1)^{k+p})=O^*(p^{O(p)})$, that is~FPT with respect to parameter~$p$, because completing a solution to satisfy all the clauses of~$C_s$ can be done in time~$O^*(2^{|C_s|})$ since \textsc{max sat-$k$} can be solved in~$O^*(2^m)$ time by Lemma~\ref{lem:pSAT}.

We now show the soundness of the algorithm.
Let~$S_0$ be an optimal solution. 
From the root of the branching tree, while it is possible, we follow a branch where all the nodes are in accordance with $S_0$. 
Let~$S_c$ be the set of variables set to \textit{true} along this branch (by definition, $S_c \subseteq S_0$), and set $S_n=S_0 \setminus S_c$.
By construction, this branch terminates at $v$, which is either a leaf or a node that deviates from $S_0$.
The leaf case being a special case of $v$ being a deviating node, we assume that $v$ deviates from~$S_0$, i.e., no child of $v$ is in accordance with~$S_0$. 
Let~$x_i$ be the variable chosen at this point by \algpSATk and consider $C_d:=C(x_i,C_u)$ that is the set of clauses not yet in~$C_s$ in which~$x_i$ appears positively or negatively. 
We know that no clause in~$C_d$ is satisfied by~$S_0$. 
Let~$x_j$ be any variable in~$S_n$. 

We claim that $S_h = (S_0 \setminus \{x_j\}) \cup \{x_i\}$ is also optimal and, by a straightforward induction, a solution at the leaves of the branching tree is as good as~$S_0$. 
Setting~$x_j$ to \textit{false}, one loses at most~$d_j(C_u)$ clauses and setting~$x_i$ to \textit{true}, one gains exactly~$d_i(C_u)$ clauses.
Indeed, we recall that no clause of~$C_d$ can be satisfied by~$S_0$, and \emph{a fortiori} by~$S_n$, since otherwise,~$v$ would not deviate from~$S_0$ (if a clause $c \in C_d$ is satisfied by~$S_0$, then the child of~$v$ that commits to satisfy~$c$ remains in accordance with~$S_0$).  
And, by our greedy choice, $d_i(C_u) \geqslant d_j(C_u)$.~\qed
\end{proof}
We may observe that 
the previous algorithm has a worse time complexity than the already known FPT algorithm for \pSC{}~\cite{blaserpartialcover}. This is rather not surprising since, as we recall, \pSC{} is a particular case of \pSAT{}  corresponding to CNFs without negative literals.

We close this section by recalling that if the length of the clauses is also part of the parameter, the decision version \textsc{sat-$k$} is FPT~\cite{niedermeier06}. In other words, denoting by \textsc{$l$-sat-$k$}, the version \textsc{sat-$k$} where each clause contains at most~$l$ literals, the following holds.
\begin{proposition}\label{alreadyknown}\cite{niedermeier06}
\textsc{$l$-sat-$k$} parameterized by $k+l$ is FPT.
\end{proposition}
The proof of Proposition~\ref{alreadyknown}, as it is given by Niedermeier~\cite{niedermeier06} works only because one has to satisfy all the clauses. The parameterized complexity of \pSAT{} with respect to $k+l$ still remains unclear and, to our opinion, deserves further investigation.

\section{Some Preliminary Thoughts About an Enhanced Weft Hierarchy: the Counting Weft Hierarchy}\label{weft}

A natural way to generalize any problem~$\Pi$ where one has to find a solution which universally satisfies a property is to define \textsc{partial}~$\Pi$, where the solution only satisfies the property a ``sufficient number of times''. 
In this sense, as mentioned, \pSC{} where one has to cover at least~$p$ elements, generalizes \SC,{} where all the elements must be covered. 
Similarly \pVC{} where one has to find a minimum subset of vertices which covers at least~$p$ edges, generalizes \VC{}, where one has to cover all the edges; yet, \mSAT{}\ generalizes~\SAT{}.
Cai studied the parameterized complexity of such partial problems (and others) in~\cite{cai}.

These partial problems come along with two parameters: the size of the solution, frequently denoted by~$k$ and the ``sufficient number of times'' quantified by~$p$. 
Many of these problems when parameterized by~$k$ are shown to be either~W[1]- or W[2]-hard, but we do not know how to prove a better membership result than the membership to W[P] (note that this is not the case of \pVC, already proved to be W[1]-complete by Guo et al.~\cite{guo2007}).
This is a quite important asymmetry between classical complexity theory as we know it from the literature (see, for example,~\cite{gj,pap,ps}) and parameterized complexity theory.

Showing the completeness of a~W[1]- or a~W[2]-hard problem, would imply that we can count up to~$p$ with a circuit of constant depth and weft~1 or~2. 
By definition of the W-hierarchy, the fact that $k$ \textit{input}-vertices of the boolean circuit can be set to true permits to deal with cardinality constraint problems, but it is not suitable to problems, such as \pSC{}, where both the value and the cardinality of the solution  are constrained. 
We sketch, in what follows, a hierarchy of circuits named \emph{counting weft hierarchy} whose classes are larger than the corresponding ones in the weft hierarchy (W-hierarchy). Basically, in the boolean circuit, we generalize the \emph{and}-vertex to a \emph{counting}-vertex. 

A counting vertex~$C_j$ with in-degree~$i$ where $j \in \{0,\ldots,i\}$ has out-degree~1 and outputs~1 iff at least~$j$ of its~$i$ inputs are~1's. Note that~$C_i$ corresponds to an \emph{and}-vertex and $C_1$ is an \emph{or}-vertex. 
A \emph{counting} circuit is a circuit with some input vertices, counting vertices, negation vertices, and exactly one output vertex. Correspondingly,~CW[$k$] is the class of problems~$\Pi$ parameterized by~$p$ such that there is a constant~$h$ and an FPT algorithm (in~$p$)~\texttt{A}, such that~\texttt{A} builds a counting circuit~$\mathcal C$ of constant depth~$h$ and weft~$k$, and $I \in \Pi$ iff $\mathcal C(I)=1$. 
It can be immediately seen that the counting weft hierarchy has exactly the same definition as the weft hierarchy up to replacing a (boolean) circuit by a (boolean) counting circuit.

Based upon the sketchy definition just above, the following can be proved by just taking the usual circuits for \SC{} and \SAT{} (recall for completeness that for \SC{}, the \textit{input}-vertices are the sets, elements are large \textit{or}-vertices taking as input the sets where they appear, and the \textit{output}-vertex is a large \textit{and}-vertex taking all or-vertices as input) and replacing the corresponding large \textit{and}-vertices by vertices~$C_p$.
\begin{proposition} 
The following inclusions hold for the counting weft hierarchy: both \pSC{} and \pSAT{} are in~CW[2].
\end{proposition}
As mentioned in the introduction, the results by Fellows et al.~\cite{fellows10} which also focus on parallel W-hierarchy with other types of vertices, cannot be used here since the counting vertices are symmetric but not bounded.

\medskip
\noindent \textbf{Acknowledgement.} The pertinent suggestions and comments of two anonymous referees have greatly improved the quality of this paper.

%


\end{document}